\documentclass[11pt,letterpaper]{article}
\usepackage[T1]{fontenc}
\usepackage{lineno}
\usepackage{amsfonts,amsmath,amsthm,pifont,amssymb,dsfont,mathtools,algpseudocode}
\usepackage[margin=1in]{geometry}
\usepackage[colorlinks,citecolor=blue,linkcolor=blue,urlcolor=red]{hyperref}
\usepackage{url}
\usepackage[linesnumbered,boxed,ruled,vlined]{algorithm2e}
\usepackage{thm-restate}
\usepackage{enumitem}
\usepackage{etoolbox}
\usepackage{graphicx}
\usepackage{mathdots}
\usepackage[normalem]{ulem}
\usepackage{xspace}
\usepackage{authblk}
\xspaceaddexceptions{]\}}
\usepackage{comment}
\usepackage{tabu}
\usepackage{framed}
\usepackage{float,wrapfig}
\usepackage{subfigure}
\usepackage{tikz}
\usepackage[draft]{fixme}
\usepackage[capitalise]{cleveref}

\newcommand*\linenomathpatch[1]{%
  \cspreto{#1}{\linenomath}%
  \cspreto{#1*}{\linenomath}%
  \csappto{end#1}{\endlinenomath}%
  \csappto{end#1*}{\endlinenomath}%
}

\linenomathpatch{equation}
\linenomathpatch{gather}
\linenomathpatch{multline}
\linenomathpatch{align}
\linenomathpatch{alignat}
\linenomathpatch{flalign}

\usetikzlibrary{arrows,arrows.meta,decorations.pathmorphing}
\theoremstyle{plain}
\newtheorem{theorem}{Theorem}[section]
\newtheorem{lemma}[theorem]{Lemma}

\newtheorem{corollary}[theorem]{Corollary}

\newtheorem{observation}[theorem]{Observation}

\theoremstyle{definition}

\fxsetup{mode=multiuser,theme=color, layout=inline}
\FXRegisterAuthor{tk}{atk}{\color{brown} {\bf Tomasz}}

\setlist[enumerate]{nosep, topsep=1ex}
\setlist[itemize]{nosep, topsep=1ex}
\setlist[description]{nosep}
\allowdisplaybreaks
\def\ShowAuthNotes{1}
\ifnum\ShowAuthNotes=1
\newcommand{\authnote}[2]{\ \\ \textcolor{red}{\parbox{0.9\linewidth}{[{\footnotesize {\bf #1:} { {#2}}}]}}\newline}
\else
\newcommand{\authnote}[2]{}
\fi

\newcommand{\rank}{\mathsf{rank}}
\newcommand{\select}{\mathsf{select}}

\newcommand{\dd}{\mathinner{.\,.\allowbreak}}

\DeclareMathOperator{\LF}{LF}
\DeclareMathOperator{\BWT}{BWT}
\DeclareMathOperator{\SA}{SA}

\DeclareMathOperator{\LCE}{LCE}
\DeclareMathOperator{\ISA}{ISA}

\DeclareMathOperator{\LFm}{\emph{modi}-LF}

\DeclareMathOperator{\COUNT}{COUNT}

\title{Compressed Inverse Suffix Arrays
}

\author{
Sharma V. Thankachan
}

\affil{
Department of Computer Science \\ North Carolina State University, Raleigh, NC, USA \\{svalliy@ncsu.edu}
}

\date{}

\begin{document}
	\maketitle
	\thispagestyle{empty}

\begin{abstract}

The suffix array ($\SA$) and inverse suffix array ($\ISA$) are fundamental data structures in string algorithms. Given a text $T[0 \dd n)$ over an integer alphabet $[0 \dd \sigma)$, the $\SA$ is a permutation of $[0 \dd n)$ that lists the starting positions of the suffixes of $T$ in lexicographic order, and the $\ISA$ is its inverse permutation. 
Explicitly storing either structure requires $\Theta(n\log n)$ bits. A central goal of \emph{compressed text indexing} is therefore to represent them in space close to the $n\log\sigma$ bits required to store the text, while supporting efficient access.
Two landmark results first addressed this challenge: the \emph{FM-index} of Ferragina and Manzini [FOCS 2000] and the \emph{Compressed Suffix Array} (CSA) of Grossi and Vitter [STOC 2000]. The FM-index typically uses $n\log\sigma$ bits
 plus lower-order terms
and supports both $\SA$ and $\ISA$ queries in roughly logarithmic time. The CSA uses $O(n\log\sigma)$ bits  and supports both queries in $O(\log_\sigma^\epsilon n)$ time, for any fixed constant $\epsilon \in (0,1]$.
These structures have since become foundational, motivating extensive work on space--time trade-offs. 
Notably, essentially all previously known approaches, including entropy-compressed and recent repetitiveness-aware encodings, use the same underlying techniques to support $\SA$ and $\ISA$ queries, yielding nearly identical space--time trade-offs for the two operations.
This raises a fundamental question: \emph{under the same asymptotic space bound, do $\SA$ and $\ISA$ queries have the same inherent query-time complexity?}

We address this question by providing strong evidence that the longstanding symmetry between $\SA$ and $\ISA$ queries need not be intrinsic to the two problems. Specifically, achieving $\log^{o(1)} n$ query time for $\SA$ under an $O(n\log\sigma)$-bit space bound would require a major breakthrough in computational geometry, providing evidence for the difficulty of improving upon the classical CSA bound. In contrast, $\ISA$ queries can be supported in doubly logarithmic time under the same space bound. The former result is a straightforward corollary of a hardness result by Chien {\it et al.}~[Algorithmica 2015]. The latter is achieved through our new encoding, which in fact uses only near-succinct space of $n\log\sigma+o(n\log\sigma)+O(n)$ bits and supports $\ISA$ queries in time
$
t_{\ISA}=O(\tau+{\log\log n}/{\log\log\sigma}),
$
where $\tau=\omega(1)$ can be an arbitrarily slowly growing function of $\sigma$.
Our structure also supports $\SA$ queries in time
$
t_{\SA}=O(t_{\ISA}\tau\log_{\sigma} n).
$
The leading $n\log\sigma$-bit term in the space bound is used solely to store the text in packed form, thereby enabling optimal-time substring extraction and efficient packed pattern searching.
Building on Sadakane's work~[SODA 2002], we also provide a suffix tree encoding that serves as an auxiliary structure to the text, uses only $O(n\log\log\log\sigma)$ bits, and supports various standard operations in polylogarithmic time. To the best of our knowledge, this is the first work to place suffix trees and suffix arrays in an auxiliary-space regime that is asymptotically smaller than the text's space for growing alphabets.

% We address this question by providing strong evidence that the longstanding symmetry between $\SA$ and $\ISA$ queries need not be intrinsic to the two problems. Specifically, achieving $\log^{o(1)} n$ query time for $\SA$ under an $O(n\log\sigma)$-bit space bound would require a major breakthrough in computational geometry, providing evidence of the difficulty of improving upon the classical CSA bound. In contrast, $\ISA$ queries can be supported in doubly logarithmic time using $n\log\sigma+o(n\log\sigma)+O(n)$ bits, substantially advancing the state of the art in near-succinct space. The former result is a straightforward corollary of a hardness result by Chien {\it et al.}~[Algorithmica 2015]. The latter is achieved through our new encoding, and our query time is
% $t_{\ISA}=O\!\left(\tau+\log\log n/\log\log\sigma\right)$, where $\tau=\omega(1)$ can be an arbitrarily slowly growing function.
% Our structure also supports $\SA$ queries in time
% $t_{\SA}=O(t_{\ISA}\tau\log_{\sigma} n)$.
% The leading $n\log\sigma$-bit term in the space bound is used solely to store the text in packed form, thereby enabling optimal-time substring extraction and efficient packed pattern searching.
% Building on Sadakane's work [Theory Comput. Syst. 2007], we also provide a suffix tree encoding that serves as an auxiliary structure to the text and supports various standard operations in polylogarithmic time using only $O(n\log\log\log\sigma)$ bits.

\end{abstract}

\newpage 
\section{Introduction and Related Work}

 Let $T[0 \dd n)$ be a string over an integer alphabet $[0 \dd \sigma)$, which we refer to as the \emph{text} throughout the paper.
Without loss of generality, we adopt the following standard assumptions: $2 < \sigma \le n$, and the last character in $T$ is a special character $\$$, which occurs exactly once and is the smallest (i.e., it is assigned the integer value $0$). 
This ensures that no suffix is a prefix of another.
The \emph{suffix array} of $T$, denoted by $\SA[0 \dd n)$, is the array that lists the starting positions of all suffixes of $T$ in lexicographic order. Since $T$ ends with $\$$, this is equivalent to listing the starting positions of all cyclic suffixes (i.e., cyclic rotations) of $T$ in lexicographic order.
The \emph{inverse suffix array}, denoted $\ISA[0 \dd n)$, is the inverse permutation of $\SA$. Therefore, $\SA[i] = j$ if and only if $\ISA[j] = i$. See Table~\ref{table_example}
 for an example. 
 Both of these arrays can be constructed in linear time~\cite{DBLP:conf/focs/Weiner73,DBLP:journals/algorithmica/Ukkonen95,DBLP:journals/siamcomp/ManberM93,DBLP:conf/focs/Farach97,DBLP:conf/icalp/KarkkainenS03,DBLP:conf/cpm/KoA03}. They form the foundation of many standard algorithms in bioinformatics, data compression, and full-text search systems~\cite{gusfield1997algorithms}. They are also widely discussed in modern textbooks~\cite{cormen2022introduction,kingsford2026introduction,Nav16BOOK}.

\begin{table}[h!]
\centering
\begin{tabular}{|c|l|c|c|c|}
\hline
{Index ($i$)} & {Sorted Suffixes} & {Sorted Cyclic Suffixes} & {$\SA[i]$} & {$\ISA[i]$} \\
\hline
0 & $T[6\dd n)=\text{\$}$        & $\text{\$banana}$  & 6 & 4 \\
1 & $T[5\dd n)=\text{a\$}$       & $\text{a\$banan}$  & 5 & 3 \\
2 & $T[3\dd n)=\text{ana\$}$     & $\text{ana\$ban}$  & 3 & 6 \\
3 & $T[1\dd n)=\text{anana\$}$   & $\text{anana\$b}$  & 1 & 2 \\
4 & $T[0\dd n)=\text{banana\$}$  & $\text{banana\$}$  & 0 & 5 \\
5 & $T[4\dd n)=\text{na\$}$      & $\text{na\$bana}$  & 4 & 1 \\
6 & $T[2\dd n)=\text{nana\$}$    & $\text{nana\$ba}$  & 2 & 0 \\
\hline
\end{tabular}
\caption{$\SA$ and $\ISA$ for the text $T[0 \dd 7) = \text{banana\$}$.}
\label{table_example}
\end{table}

% \begin{table}[h!]
% \centering
% \begin{tabular}{|c|l|c|c|c|}
% \hline
% {Index ($i$)} & {Sorted Suffixes} & {Sorted Cyclic Suffixes} & {$\SA[i]$} & {$\ISA[i]$} \\
% \hline
% 0 & $T[6\ldots n)=\text{\$}$        & $\text{\$banana}$  & 6 & 4 \\
% 1 & $T[5\ldots n)=\text{a\$}$       & $\text{a\$banan}$  & 5 & 3 \\
% 2 & $T[3\ldots n)=\text{ana\$}$     & $\text{ana\$ban}$  & 3 & 6 \\
% 3 & $T[1\ldots n)=\text{anana\$}$   & $\text{anana\$b}$  & 1 & 2 \\
% 4 & $T[0\ldots n)=\text{banana\$}$  & $\text{banana\$}$  & 0 & 5 \\
% 5 & $T[4\ldots n)=\text{na\$}$      & $\text{na\$bana}$  & 4 & 1 \\
% 6 & $T[2\ldots n)=\text{nana\$}$    & $\text{nana\$ba}$  & 2 & 0 \\
% \hline
% \end{tabular}
% \caption{$\SA$ and $\ISA$ for the text $T[0 \dd 7) = \text{banana\$}$.}
% \label{table_example}
% \end{table}

\vspace{-7mm}

\subsection{Space Regimes and Known Trade-offs}
We now review the main known space--time trade-offs for the suffix array (SA) and inverse suffix array (ISA).
Throughout this paper, $\epsilon$ denotes an arbitrarily small constant satisfying $0 < \epsilon < 1$.
Also, $\log x$ is defined as $\max\{1,\log_2 x\}$, and iterated logarithms use this clamped log recursively.

\vspace{-3mm}
\paragraph{Linear Space.}
Both SA and ISA can be stored explicitly using $O(n)$ words, 
which is equivalent to $\Theta(n \log n)$ bits. This space is typically referred to as \emph{uncompressed} or \emph{linear space}.  
On the other hand, the text itself can be represented in  $n\lceil \log \sigma \rceil$ bits in a \emph{packed setting}, where each character in $[0 \dd \sigma)$ is encoded using $\lceil \log \sigma \rceil$ bits.\footnote{Strictly speaking, it suffices to store $T[0 \dd n-1) \in [1 \dd \sigma)^{n-1}$
%in $(n-1)\lceil \log (\sigma -1) \rceil$ bits, 
since $\$$ is unique and its position is fixed, and the information-theoretic lower bound is $(n-1)\log(\sigma-1)$ bits. However, for convenience of presentation, we assume that the text $T$ itself is stored using a fixed-length encoding of $\lceil \log \sigma \rceil$ bits per character. This simplifying assumption is without loss of generality, since all structures developed in this work are auxiliary.}
Consequently, compared to the text, the plain representations of SA and ISA  can incur a logarithmic-factor blowup in space. This overhead is particularly critical in application domains such as bioinformatics, where $n$ can be on the order of billions, while the alphabet size is  constant.
This has motivated the development of their space-efficient encodings.

\vspace{-4mm}

\paragraph{Succinct and Compact Space.}
We say that a data structure uses \emph{succinct space}
if it occupies $n\log\sigma + o(n\log\sigma)$ bits. We say that it uses \emph{compact space} if it
occupies $O(n\log\sigma)$ bits.
The first set of major breakthroughs in space-efficient encodings of SA and ISA emerged in the year 2000 with the introduction of the \emph{Compressed Suffix Array} (CSA)~\cite{DBLP:journals/siamcomp/GrossiV05} and the \emph{FM-index}~\cite{DBLP:journals/jacm/FerraginaM05}.
The paper~\cite{DBLP:journals/siamcomp/GrossiV05} presents two space--time trade-offs, and a third trade-off can be obtained based on Rao's work~\cite{DBLP:journals/ipl/Rao02} (originally presented for a binary alphabet), 
as summarized in Table~\ref{table_CSA}.
The FM-index provides a solution 
by leveraging the Burrows--Wheeler Transform (BWT)~\cite{burrows1994block} of the text. Specifically, it stores the $\BWT$ in a succinct data structure supporting \emph{rank} and \emph{select} operations, together with sampled versions
of SA and ISA. These samples require $(n/\tau)\log\sigma$ bits for a parameter $\tau \geq 1$, resulting in a total space of $(1 + 1/\tau+o(1))\, n \log \sigma$ bits, which can be made succinct by tuning the parameter $\tau$.
Both SA and ISA queries are answered using $O(\tau\log_{\sigma} n)$ rank
queries on the BWT. The  version of the FM-index by Belazzougui and Navarro supports SA/ISA queries in $O(\tau\log_{\sigma} n)$ time~\cite{DBLP:journals/talg/BelazzouguiN14}, using $(1+1/\tau)n\log\sigma+O(n\log\log\sigma)$ bits of space.

\begin{table}[h]
\centering
\begin{tabular}{|l|l|l|}
\hline
\textbf{Reference} & \textbf{Space (in bits)} & \textbf{Query Time} \\
\hline
\cite{DBLP:journals/siamcomp/GrossiV05} & $O(n \log \sigma)$ & $O(\log_{\sigma}^{\epsilon} n)$ \\
\hline
\cite{DBLP:journals/siamcomp/GrossiV05} & $O(n \log \sigma \cdot \log \log_{\sigma} n)$ & $O(\log \log_{\sigma} n)$ \\
\hline
\cite{DBLP:journals/ipl/Rao02} & $O(n \log \sigma \cdot \log_{\sigma}^{\epsilon} n)$ & $O(1)$ \\
\hline
\end{tabular}
\caption{Space--time trade-offs for SA/ISA queries.}
\label{table_CSA}
\end{table}

These foundational results initiated a substantial body of research in this area.
See~\cite{DBLP:journals/mst/Sadakane07,DBLP:conf/soda/GrossiGV03,DBLP:conf/soda/FerraginaM04,DBLP:journals/ipl/Rao02,DBLP:journals/jacm/GagieNP20,DBLP:journals/talg/BelazzouguiN14,DBLP:conf/focs/KempaK23,DBLP:conf/soda/KempaK26} for some key developments.
Among these, Sadakane's work on \emph{Compressed Suffix Trees}~\cite{DBLP:journals/mst/Sadakane07} is particularly notable.
It essentially shows that, with an additional $O(n)$ bits of space, one can support various suffix-tree functionalities.
% , including the fundamental \emph{Longest Common Extension} (LCE) queries.
% An LCE query $\LCE(i,j)$ on $T$ asks to return the length of the longest common prefix of the suffixes of $T$ starting at positions $i$ and $j$. 
We refer to~\cite{Nav16BOOK} for further reading.

We next describe compressed space. We remind the reader that the results
introduced in this work are not primarily concerned with achieving
compressed space itself; rather, compressed space provides a useful context
for understanding our results.

\vspace{-4mm}
\paragraph{Compressed Space.}
Text compression schemes can be roughly categorized into two types: \emph{entropy-based} compression and \emph{repetitiveness-aware} compression.
Entropy-based schemes have witnessed many early successes, leading to
higher-order entropy-compressed versions of the FM-index~\cite{DBLP:conf/soda/FerraginaM04,DBLP:journals/talg/BelazzouguiN14} and the CSA~\cite{DBLP:conf/soda/GrossiGV03}.
However, 
the theoretical foundations of {repetitiveness-aware} compression schemes have long remained difficult to fully understand. Only in recent years has the research community begun to systematically investigate this area and achieve substantial progress; see the  surveys by Navarro~\cite{DBLP:journals/csur/Navarro21a,DBLP:journals/csur/Navarro21b}.
These schemes are particularly well suited for texts that contain many long repeated substrings. Such repetitiveness is common in modern datasets, including collections of multiple genomes (e.g., pangenomes), versioned documents, and other structured data.
Popular repetitiveness-aware compression methods include Lempel--Ziv compression~\cite{ziv1977universal}, run-length encoded BWT~\cite{burrows1994block}, and grammar compression. While there has been notable progress in repetitiveness-aware text indexing data structures, primarily for pattern matching, only recently has a clearer understanding emerged of how these seemingly disparate compression schemes are theoretically connected.
We now highlight several key results, including recent developments in repetitiveness-aware encodings of $\SA$ and $\ISA$.

\begin{enumerate}
    \item An $O(r \log(n/r))$-word data structure with $O(\log(n/r))$ query time by Gagie, Navarro and Prezza~\cite{DBLP:journals/jacm/GagieNP20}. Here, $r$, a popular measure of repetitiveness, denotes the number of runs in the text's BWT, where a run is a maximal sequence of identical characters. 
  Although the measure $r$ tends to be small compared to $n$ for highly repetitive texts in practice, at that time it was unclear from a theoretical standpoint how $r$ relates to other popular measures of repetitiveness, such as $z$, the number of phrases in the Lempel-Ziv parsing of the text. However, soon after, a number of foundational results appeared that ultimately unified the landscape of repetitiveness-aware compression under a new measure, $\delta$, called the \emph{substring complexity}~\cite{DBLP:journals/algorithmica/RaskhodnikovaRRS13,DBLP:journals/tit/KociumakaNP23}, defined as $\delta = \max_t d_t(T)/t$, where $d_t(T)$ denotes the number of distinct substrings of length $t$ in $T$.
The measure $\delta$ serves as a lower bound for $z$, $r$, and various other repetitiveness measures. 
From an upper bound perspective, we have $z = O(\delta \log (n/\delta))$ and $r = O(\delta \log \delta \log (n/\delta))$~\cite{DBLP:journals/cacm/KempaK22,DBLP:conf/stoc/KempaP18}.
This shows that these measures are all within polylogarithmic factors of each other, despite appearing to be quite different. 
Therefore, in terms of $\delta$, the space of the encoding  in~\cite{DBLP:journals/jacm/GagieNP20} becomes $O(\delta \log \delta \log^2(n/\delta))$ words.

\item  It has been shown that $O(\delta \log (n/\delta))$ words represent the \emph{optimal repetitiveness-aware space} to encode $T$, meaning this amount of space (i.e., in terms of $n$ and $\delta$) is both necessary and sufficient~\cite{DBLP:journals/tit/KociumakaNP23}. 
Kempa and Kociumaka~\cite{DBLP:conf/focs/KempaK23} then showed that it is possible to achieve an encoding within this space while supporting both $\SA$ and $\ISA$ queries in $O(\log^{4+\epsilon} n)$ time.

\item In a recent paper~\cite{DBLP:conf/soda/KempaK26}, Kempa and Kociumaka showed that a slightly faster query time of $O(\log n / \log \log n)$ is possible for both $\SA$ and $\ISA$ queries by increasing the space to $O(\delta \log^{4+\epsilon} n)$. They also proved a matching lower bound, showing that any data structure using $O(\delta \log^{O(1)} n)$ space must incur a query time of $\Omega(\log n / \log \log n)$ for both $\SA$ and $\ISA$.

\end{enumerate}
%%%
As observed, all these existing works, across different space regimes, achieve the same asymptotic space and time complexities for $\SA$ and $\ISA$ queries, including the lower bound mentioned above.
This raises the following fundamental question that we address in this paper.
 \textit{Given the same (asymptotic) space budget, do $\SA$ and $\ISA$ queries exhibit the same query time complexity?}

\vspace{-4mm}

\subsection{Hardness of $\SA$ Queries in Compact Space}
 Note that the work by Kempa and Kociumaka~\cite{DBLP:conf/soda/KempaK26} resolves  the above question for $O(\delta \log^{O(1)} n)$ space.
In another  work~\cite{DBLP:conf/soda/KempaK26a}, they proved an equivalence between several string problems (including SA/ISA queries) and certain types of \emph{prefix range queries}, a newly introduced class of problems that appears more basic. They established this equivalence across four aspects: \emph{space usage, query time, construction time, and construction working space}. Rather than viewing these string problems from a hardness perspective, their work suggests that one can instead focus on improving these simpler abstract queries to achieve further advances in fundamental string problems.
%Despite these abstractions, the CSA trade-offs listed in Table~\ref{table_CSA} have remained unimproved.

Another relevant work is a reduction by Chien {\it et al.}~\cite{DBLP:journals/algorithmica/ChienHSTV15} connecting the \emph{two-dimensional orthogonal range reporting} problem to the text indexing problem. The text indexing problem asks to preprocess a text into a data structure that can report all occurrences of a query pattern as substrings of the text. Since the suffix array (together with the text) can be used to support pattern matching, the following lemma can be easily obtained from their reduction.
%; therefore, we sketch the proof in Appendix~\ref{sec_appendix} for completeness.

\begin{restatable}{lemma}{hardnesslemma}
\label{lem_reduction}
If there exists an $O(n)+S(n)$-bit data structure supporting $\SA$ queries in $Q(n)$ time for a text of length $n$ over a constant alphabet, then one can construct an $O(n\log n)+S(\Theta(n\log n))$-bit structure for two-dimensional range reporting over $n$ points in an $[n]\times[n]$ grid, supporting queries in $O(\log^4 n+k\cdot Q(n))$ time, where $k$ denotes the output size. We may safely assume that $Q(n)=O(\log n)$, since an $O(\log n)$-time solution using $O(n)$ bits is already known to exist.
\end{restatable}

Applying the existing trade-offs for $\SA$ queries in Table~\ref{table_CSA} to this reduction yields the following trade-offs for two-dimensional range reporting.

\begin{enumerate}
 \item An $O(n)$-word data structure with query time $O(\log^4 n + k  \log^{\epsilon} n)$.
  \item An $O(n\log\log n)$-word data structure with query time $O(\log^4 n + k  \log\log n)$.
   \item An $O(n\log^{\epsilon}n)$-word data structure with query time $O(\log^4 n + k )$.
\end{enumerate}
These bounds match those of the classic results of Chazelle~\cite{DBLP:journals/siamcomp/Chazelle88}, except that the $\log^4 n$ term appearing here is replaced by a smaller term in his results. Chazelle's result has witnessed several improvements
~\cite{agarwal1999geometric,DBLP:conf/focs/AlstrupBR00,DBLP:conf/wads/Nekrich07,DBLP:conf/compgeom/ChanLP11}, all of which have focused on reducing the initial term; the per-output reporting time has remained unchanged.
Therefore, improving any of the bounds in Table~\ref{table_CSA} for $\SA$ queries would lead to improved per-output reporting time for two-dimensional range reporting, indicating that substantially faster $\SA$ queries are unlikely without a breakthrough in computational geometry.

\vspace{-3mm}
 \subsection{Our Results}

We assume the word RAM model of computation with word size $\Theta(\log n)$. 
In this model, any block of $b \le \log n$ bits fits within $O(1)$ words and can therefore be accessed and interpreted as an integer in the range $[0 \dd 2^b)$ in constant time. Consequently, our text $T[0 \dd n)$ can be stored 
in a \emph{packed form}
using $n \lceil \log \sigma \rceil$ bits,  where each character is represented
using a fixed-length encoding of $\lceil\log\sigma\rceil$ bits. Then any substring of $O(\log_{\sigma} n)$ characters fits within $O(1)$ machine words and can be accessed in constant time. 
The following is our main result.

\begin{theorem}[\bf ISA and SA Encodings]\label{thm_main}
Let $\tau$ be an integer-valued parameter and $\epsilon$ an arbitrarily small fixed  constant such that $1\leq\tau=O(\log\sigma)$ and $0<\epsilon<1$.
Suppose the text $T[0 \dd n)$ over the integer alphabet $[0 \dd \sigma)$ is
stored  in packed form, such that any substring of length
$O(\log_{\sigma} n)$ fits in $O(1)$ machine words and can be accessed in
$O(1)$ time. 
We can maintain an auxiliary data structure with the following space--time trade-offs for inverse suffix array $(\ISA)$ queries.

\begin{itemize}
    \item $O((n/\tau)\log\sigma + n\log^{\epsilon}\sigma)$ bits  and
    query time $t_{\ISA}=O(\tau+\log\log n/\log\log\sigma)$;
\item $O((n/\tau)\log\sigma + n\log\log\sigma)$ bits and
 $t_{\ISA}=O(\tau+\log\log_{\sigma} n)$;
\item     $O((n/\tau)\log\sigma + n\log\log\log\sigma)$ bits and
  $t_{\ISA}=O((\tau+\log\log\log \sigma)\log\log\sigma+ \log\log n)$.
\end{itemize}
In all cases, suffix array $(\SA)$ queries can be supported in time
$t_{\SA}=O(t_{\ISA}\tau\log_{\sigma} n)$.

\end{theorem}

The packed  text takes $n \lceil \log\sigma\rceil =n\log\sigma+O(n)$ bits.\footnote{We could also use higher-order entropy-compressed  representations
 supporting efficient substring retrieval~\cite{DBLP:conf/soda/FerraginaV07}.} 
By setting $\tau=\omega(1)$ to be an arbitrarily slowly growing function of $\sigma$, we can achieve three trade-offs while keeping the space complexity at $n\log\sigma+o(n\log\sigma)+O(n)$ bits (which is succinct when $\sigma$ is a growing function of $n$): $t_{\ISA}$ can be made doubly logarithmic, while $t_{\SA}$ remains near-logarithmic, as in the FM-index.
We also prove a lower bound that justifies the $\tau$ term in the $\SA$ query time over a broad range of parameters.

\paragraph{Suffix Tree and LCE Queries.}
From the work of Sadakane~\cite{DBLP:journals/mst/Sadakane07} (see also~\cite{russo2011fully}), it follows immediately that, by maintaining additional structures occupying a total of $O(n)$ bits (so the overall space complexity remains unchanged), we can obtain an encoding of the \emph{suffix tree}~\cite{DBLP:conf/focs/Weiner73} that supports a variety of functionalities. This includes supporting \emph{Longest Common Extension} (LCE) queries in time $O(t_{\SA}+t_{\ISA})$, where $\LCE(i,j)$ returns the length of the longest common prefix of the suffixes starting at positions $i$ and $j$ (see~\cite{DBLP:conf/stoc/KempaK19} for a compact space structure with $O(1)$ query time).
%\footnote{The $O(n\log\sigma)$-bit data structure of~\cite{DBLP:conf/stoc/KempaK19} supports $\LCE$ queries in $O(1)$ time.}
We refer to Table 1 in~\cite{russo2011fully} for a comprehensive list of other supported functionalities.
Thus, our most space-efficient encoding of the suffix tree can be obtained by setting
$\tau = \lceil \log\sigma/\log\log\log\sigma \rceil $ in the third ISA trade-off.
%(and simplifying).

\begin{corollary}[\bf Suffix Tree  in Just $O(n\log\log\log\sigma)$ Extra Bits]\label{cor:CST}
Suppose the text $T[0 \dd n) \in [0 \dd \sigma)^n$ is stored separately in packed form. Then, an encoding of its suffix tree  supporting various basic operations in polylogarithmic time can be obtained using $O(n\log\log\log\sigma)$ additional bits. 
%In particular, it supports $\ISA$ queries in $O(\log\sigma+\log\log n)$ time, and $\SA$ and $\LCE$ queries in $O(\log n(\log\sigma+\log\log n))$ time.
\end{corollary}

\paragraph{Optimal Time Substring Extraction.}
Since the  text itself is available in packed form, any substring of length $\ell$ can be extracted in optimal  $O(1+\ell/\log_{\sigma} n)$ time. 
In a standard FM-index, the (LF-based) substring extraction takes
$O(t_{\ISA}+\ell)$ time.
While the CSA can achieve this time bound, its leading space term is
$(1+\epsilon^{-1})n\log\sigma \geq 2n\log\sigma$ bits since $0<\epsilon\le 1$.

\paragraph{Packed Pattern Searching in (Near-)Succinct Space.}
The support for $\SA$, $\ISA$, and $\LCE$ queries also enables efficient pattern matching. Given a pattern $P$, a \emph{counting query} returns the number of occurrences of $P$ in $T$, while a \emph{reporting query} returns all such occurrences, i.e., their starting positions. The standard procedure for both types of queries is to first find the range $[sp\dd ep]$ of $P$, called the \emph{suffix range}, where $\mathrm{OCC}=\{\SA[k]\mid sp\leq k\leq ep\}$ is the set of occurrences of $P$ in $T$. The suffix range (and $\mathrm{OCC}$) is empty if $P$ has no occurrences.
The counting query then returns $|\mathrm{OCC}|$, while the reporting query outputs $\mathrm{OCC}$. 

\begin{corollary}
    For any pattern $P$ given in packed form, our data structure can find its suffix range in time $O(|P|/\log_{\sigma} n+\log n+t_{\SA}\log\log n)$ and then report its occurrences in $t_{\SA}$ time per output.
\end{corollary}

\begin{proof}
    The standard binary search takes $O((|P|/\log_{\sigma} n+t_{\SA})\log n)$ time, which can be improved to $O(|P|/\log_{\sigma} n+t_{\SA}\log n)$ using $\LCE$ queries. To achieve the claimed complexity, we sample every $\log n$-th $\SA$ entry and store these entries explicitly, along with a sampled suffix tree that allows the $\LCE$ of any two sampled suffixes to be computed in $O(1)$ time using $O(n)$ additional bits. A primary binary search over the sampled entries reduces the search range to $O(\log n)$, after which a second binary search requires only $O(\log\log n)$ $\SA$ queries.
\end{proof}
Note that the FM-index uses the so-called \emph{backward search} procedure,
which processes the pattern one character at a time and thus takes $O(|P|)$
steps. While the CSA supports efficient packed pattern searching, its leading
space term is at least $2n\log\sigma$ bits. To the best of our knowledge, no
previous full-text index supports packed pattern searching in
$O(|P|/\log_{\sigma} n+\log^{O(1)} n)$ time while using space with
coefficient $1$ on the leading $n\log\sigma$ term.

% To the best of our knowledge, no previous text index achieves \emph{packed pattern-searching} time $O(|P|/\log_{\sigma} n+\log^{O(1)} n)$ for counting while using $(1+o(1))n\log\sigma$ bits of space.

\paragraph{Complexity Separation.}
One implication of Theorem~\ref{thm_main} is that it reveals a difference in the space--time trade-offs for $\SA$ and $\ISA$ queries across space regimes. Our result shows that $\ISA$ queries can be supported in doubly logarithmic time using compact space, whereas Lemma~\ref{lem_reduction} highlights the difficulty of achieving improved bounds for $\SA$ queries. This leads to the following.

\begin{corollary}
    In the compact-space regime, there exists a separation between the time complexities of $\SA$ and $\ISA$ queries, unless a linear-space data structure for 2D range reporting on $n$ points with query time $O(\log^4 n+k\log\log n)$ exists, where $k$ is the output size. 
\end{corollary}
Recall that no such complexity separation exists in repetitiveness-aware space~\cite{DBLP:conf/soda/KempaK26}.

\paragraph{Roadmap.}
Section~\ref{sec:preli} presents the preliminaries. In Section~\ref{sec:3}, we first describe an implementation of the FM-index and then introduce a simple modification in which the BWT is replaced by the text itself while still supporting $\ISA$ queries. Section~\ref{sec:4} presents the required building-block data structures. Section~\ref{sec:5} derives our final results through a modification of a CSA trade-off. 
Section~\ref{sec:LB} presents our lower-bound result.
Section~\ref{sec:6:conc} concludes the paper.
%The  proof of Lemma~\ref{lem_reduction} is included in Appendix~\ref{sec_appendix}.

\section{Preliminaries}\label{sec:preli}

We will follow $0$-based indexing.
For any string $T[0 \dd n)$, let $T[i \dd j]$, with $i \leq j$, denote the substring starting at index $i$ and ending at index $j$.
For $j > i$, let $T[j \dd i]$ denote the cyclic substring obtained by concatenating $T[j \dd n)$ and $T[0 \dd i]$.
As we are dealing with an integer alphabet $[0 \dd \sigma)$, we use the usual $<$ relation to denote lexicographic order between two characters (i.e., $x < y$ iff $x$ is lexicographically smaller than $y$).

We represent each character using $\lceil \log \sigma \rceil$ bits.
We  view a string $S$ of $k$ characters over $[0 \dd \sigma)$ as an integer in
$[0 \dd 2^{k\lceil \log \sigma\rceil})$,
whose base-$2^{\lceil \log \sigma\rceil}$ representation corresponds to $S$.
Since all characters use the same number of bits, this representation preserves the lexicographic order of the strings.

\subsection{Rank and Select Queries}
Let $S[0 \dd n)$ be a string  over an alphabet $[0 \dd \pi)$, where $\pi =n^{O(1)}$. A rank operation $\rank_S(i,c)$, where $i \in [0 \dd n)$ and $c \in [0 \dd \pi)$, returns the number of occurrences of the character $c$ in  $S[0\dd i]$. A \emph{select} operation, $\select_S(j,c)$, returns the index in $S$ (if it exists) of the $j$-th occurrence of  $c$. The number of characters in $S$ that are strictly smaller than $c \in [0 \dd \pi)$ is denoted by $\COUNT_S[c]$.

That is,
\[
\COUNT_S[c] = \bigl|\{i \mid i  \in [0 \dd n) \text{ and } S[i] < c\}\bigr|.
\]
%%%
The following results consider the basic case where $S$ is a binary sequence with $m$ ones.

\begin{itemize}

\item {\bf Succinct Bit Vector:} This is an  $n+o(n)$-bit  data structure that   supports both $\rank$ and $\select$ operations in $O(1)$ time~\cite{DBLP:conf/focs/Jacobson89,DBLP:conf/soda/ClarkM96}.
   
   \item 
   \emph {\bf Indexable Dictionary:} This data structure uses $m \log(n/m) + O(m)$ bits, which is close to the information-theoretic lower bound of $\log \binom{n}{m}$ bits required to represent a binary string with $m$ ones~\cite{DBLP:journals/talg/RamanRS07}. While it does not support all operations in $O(1)$ time, it can efficiently support two specific operations in $O(1)$ time: 
   
   \begin{enumerate}
       \item $\select_S(j,1)$ for any $j$.
       \item $\rank_S(i,1)$ given that $S[i] = 1$. This operation is referred to as \emph{partial rank}.
   \end{enumerate}
We refer to recent related work for further developments on this topic~\cite{DBLP:conf/icalp/LiangZ25}.

    \item  {\bf Monotone Minimal Perfect Hash Function (MMPHF):}
This data structure requires only $O(m\log \log (n/m))$ bits, 
which can be asymptotically smaller than the space required to represent the bit vector itself.
 It operates independently of the original data~\cite{DBLP:conf/soda/BelazzouguiBPV09,DBLP:conf/alenex/BelazzouguiBPV09,DBLP:journals/talg/BelazzouguiN14} (i.e., $S$ is needed only during construction) and supports the following useful operation in $O(1)$ time: given an input $i$, it correctly returns $\rank_S(i,1)$ if $S[i]=1$, and an arbitrary value otherwise.
 Note that the data structure cannot determine whether $S[i]=1$; rather, it supports $\rank_S(i,1)$ under the promise that $S[i]=1$.
% There also exists a more space-efficient version of MMPHF, using $O(m\log \log \log (n/m))$ bits and supporting queries in $O( \log \log (n/m))$ time.
\end{itemize}

\subsection{The Burrows-Wheeler Transform and LF-Mapping}\label{sec:BWT}
A \emph{cyclic suffix} of $T[0 \dd n)$ is obtained by taking the suffix $T[i \dd n)$ and appending the prefix $T[0 \dd i)$ at the end. Formally, the $i$-th cyclic suffix, denoted $T_i$, is $T[i \dd n) \circ T[0 \dd i)$, %for $i > 0$, and $T_0 = T$, 
where $\circ$ denotes concatenation.
The \emph{$\BWT$-matrix} of $T$ is the $n \times n$ matrix $M[0 \dd n)[0 \dd n)$ whose rows consist of all cyclic suffixes of $T$, sorted lexicographically.
The \emph{Burrows--Wheeler Transform} (BWT) of $T$, denoted $\BWT[0 \dd n)$, is the string formed by taking the last column of this matrix. That is, $\BWT[0 \dd n) = M[0][n-1] \circ M[1][n-1] \circ M[2][n-1]\circ \cdots \circ M[n-1][n-1]$. 
The LF-mapping is a key operation on the BWT, defined as 
$\LF[i] = \ISA[(\SA[i] - 1) \bmod n]$. Although this definition involves $\mathrm{SA}$ and $\mathrm{ISA}$, for any $0 \leq i,j < n$, we can decide the order between $\mathrm{LF}[i]$ and $\mathrm{LF}[j]$ by simply comparing $\mathrm{BWT}[i]$ and $\mathrm{BWT}[j]$:  
\[
\mathrm{LF}[j] < \mathrm{LF}[i]
\iff
\big(\mathrm{BWT}[j] < \mathrm{BWT}[i]\big)
\;\lor\;
\big((\mathrm{BWT}[j] = \mathrm{BWT}[i]) \land (j < i\big)).
\]
This yields the following well-known expression for computing $\LF$ from  $\BWT$:

$$\LF[i] = |\{j \mid \BWT[j] < \BWT[i]\}| + |\{ j \mid j \in [0 \dd i] \text{ and } \BWT[j] = \BWT[i]\}|-1. $$
Here $-1$ accounts for 0-based indexing. 
Let $c = \BWT[i]$. Then we can write 
\[
\LF[i] = \COUNT_{\BWT}[c] + \rank_{\BWT}(i, c) - 1.
\]
Also, let $\LF^0[i]=i$ and 
$\LF^k[i]=\LF[\LF^{k-1}[i]]=\ISA[(\SA[i]-k)\bmod n]$ for any integer $k>0$.

\section{ISA Encodings with  Logarithmic Query Time}\label{sec:3}

\subsection{BWT-Based Approach} \label{sec:FMI}
In the standard approach based on the FM-index~\cite{DBLP:journals/jacm/FerraginaM05}, we first maintain an efficient data structure that supports $\LF$-mapping. We then replace $\ISA$ with a sampled variant (called sampled ISA), defined as follows.
Let $g$ be an integer parameter, which we call the \emph{sampling factor}. A sampled ISA explicitly stores $\ISA[y]$ for all text positions $y \in [0 \dd n)$ such that $y$ is divisible by $g$, to support $O(1)$-time lookup for such positions.
%Additionally, store the same for the rightmost text position (i.e., $\ISA[n-1])$.
Then, given a query index $x$, we obtain $\ISA[x]$ as follows:

\begin{enumerate}
    \item If $\ISA[x]$ is stored explicitly, return it in $O(1)$ time.
    
    \item Otherwise, let $y$ be the smallest multiple of $g$ such that $y>x$. If $y\leq n-1$, then $\ISA[y]$ is stored. Let $(i,\SA[i])=(\ISA[y],y)$ and $t=y-x<g$. It follows that the answer is $\LF^t[i]$. If $y>n-1$, set $y=0$. Then let $(i,\SA[i])=(\ISA[0],0)$ and $t=n-x<g$; it follows that the answer is $\LF^t[i]$.
\end{enumerate}

Therefore, the problem reduces to efficiently supporting iterative 
$\LF$-mapping, which in turn requires storing $\COUNT_{\BWT}$ and providing 
efficient rank operations on the $\BWT$.
The array $\COUNT_{\BWT}$ can be stored in $O(\sigma+\sigma \log(n/\sigma))$ bits 
using Elias--Fano encoding while supporting $O(1)$-time lookup.
To support efficient rank operations, several methods have been developed for 
succinctly representing the underlying string; see~\cite{DBLP:conf/soda/GrossiGV03,DBLP:journals/jda/Navarro14} for wavelet trees supporting 
$O(\log \sigma)$-time operations, and~\cite{DBLP:conf/soda/GolynskiMR06,DBLP:journals/talg/BelazzouguiN15} for representations supporting 
$O(\log\log \sigma)$-time operations.
Moreover, under $O(n\log^{O(1)} n)$-bit space, $O(1)$-time rank operations are not possible, unless 
$\sigma=\log^{O(1)}n$~\cite{DBLP:journals/talg/BelazzouguiN15}.
However, in our setting, the operation 
$\rank_{\BWT}(i,\BWT[i])$ is more restricted than a general rank operation.
Belazzougui and Navarro (see Section~3 of~\cite{DBLP:journals/talg/BelazzouguiN14}) 
exploited this property and showed that it can be supported in constant time 
by maintaining an $O(n\log\log\sigma)$-bit data structure in addition to the 
$\BWT$. Thus, $\LF$-mapping can be supported in $O(1)$ time. Consequently,
$\ISA$ queries can be answered in $O(g)$ time using 
$n\log\sigma +O(n \log\log\sigma)+ O((n/g)\log n)$
bits of space. Typically, $g$ is chosen as 
$\omega(\log_{\sigma} n)$.
As the details of the aforementioned $O(n\log\log\sigma)$-bit data structure 
are important for our purposes, we present Lemma~\ref{lem_COUNTplusRANK} (and Observation~\ref{obs:COUNT+RANK}), which simplifies the ideas from~\cite{DBLP:journals/talg/BelazzouguiN14}.

\begin{lemma}[``$\COUNT + \rank$''  Queries] \label{lem_COUNTplusRANK}
Let  $S[0 \dd n) \in [0 \dd \pi)^n$ be a string, where $\pi = n^{O(1)}$. Suppose $S$ is represented in a way that allows constant-time random access to its characters. Then we can associate it with an $O(n \log \log \pi)$-bit data structure such that, when $i$ is given as a query, we can in $O(1)$ time return $$\COUNT_S[S[i]] + \rank_S(i, S[i])-1.$$
%Alternatively, we can associate an $O(n\log\log\log\pi)$-bit data structure and support the same queries in $O(\log\log \pi )$ time.
Note that when $S=\BWT$ and $c=\BWT[i]$, the value returned is $\LF[i]$.
\end{lemma}

\begin{proof}
For $k \in [0 \dd \pi)$, let $B_k[0 \dd n)$ denote the bit vector such that $B_k[j]=1$ if and only if $S[j]=k$, and let $B$ be the concatenation of $B_0,B_1,B_2,\dots,B_{\pi-1}$ in this order. 
An MMPHF over $B$ is our data structure in space 
\(O(|S|\log\log(|B|/|S|)) \subseteq O(n\log\log\pi)\) bits.
Given a query $i$, we first obtain $c=S[i]$ by accessing $S$. 
Note that $B[nc+i]=1$, and $\rank_B(nc+i,1)-1$ gives our answer. This rank query can be computed correctly
using the MMPHF over $B$ in $O(1)$ time.
The approach in~\cite{DBLP:journals/talg/BelazzouguiN14} differs in that it separately maintains an MMPHF for each $B_k$ and $\COUNT_S$.
% With the first version of the MMPHF, the space is
% \(O(|S|\log\log(|B|/|S|)) \subseteq O(n\log\log\pi)\) bits and the query time is \(O(1)\).
% With the second version, the space is
% \(O(|S|\log\log\log(|B|/|S|)) \subseteq O(n\log\log\log\pi)\) bits, while the query time is
% \(O(\log\log(|B|/|S|)) \subseteq O(\log\log\pi)\).
\end{proof}

\begin{observation}\label{obs:COUNT+RANK}
The data structure in Lemma~\ref{lem_COUNTplusRANK} uses $S$ only for retrieving $S[i]$ given $i$. Hence, it can be implemented without storing $S$ when the input 
is provided as $(i,S[i])$ instead of only $i$.    
\end{observation}

\subsection{BWT-Free (Text-Based) Approach}\label{sec:free}
In the above implementation of the FM-index, we can eliminate the need to store the $\BWT$ by using the text directly while still supporting $\ISA$ queries as follows. Letting $i_0=i$, where $i =\ISA[y]$, 
%our goal is to
we 
iteratively compute $i_1=\LF[i_0], i_2=\LF[i_1], i_3=\LF[i_2],\ldots, i_t=\LF[i_{t-1}]$.
%, starting from $(i_0,\SA[i_0]) =(\ISA[y],y)$.
Since $\SA[i_0] =y$ is available initially, we can compute $\SA[i_k]=\SA[\LF^k[i_0]]=((\SA[i_0]-k)\bmod n)$ for any $k$ in constant time. Therefore, for $k=0,\ldots,t-1$, we compute $i_{k+1}=\LF[i_k]$ using $(i_k,\SA[i_k])$ as follows: first, retrieve $\BWT[i_k]=T[(\SA[i_k]-1)\bmod n]$ directly from the text (instead of $\BWT$). Then, using Observation~\ref{obs:COUNT+RANK}, we query the data structure in 
Lemma~\ref{lem_COUNTplusRANK} (built over the $\BWT$) with input 
$(i_k,\BWT[i_k])$ to obtain $i_{k+1}$.
The space--time complexities of this alternative approach remain the same 
as before.
However, the benefits will become apparent as we build upon this  idea. 

For notational convenience, let $\LFm$ denote the (modified LF) function mapping a pair $(j,\SA[j])$ to $(j',\SA[j'])$, where $j'=\LF[j]$, and let $\LFm^k$ denote the result of applying $\LFm$ iteratively $k\geq 0$ times.
Formally,

\[
\LFm^k[(j, \SA[j])] = (\LF^k[j], \SA[\LF^k[j]]) = (\LF^k[j], (\SA[j]-k) \bmod n).
\]
Therefore, an ISA query ultimately reduces to computing $\LFm^t[(i, \SA[i])]$ from  input $(i,\SA[i],t)$, such that $t<g$ (the sampling factor) and $\SA[i]$ is divisible by $g$.

\paragraph{Note.} The standard (FM-index based) procedure for \(\SA\) value decoding is also based on iterative
\(\LF\)  until a sampled position containing the answer is reached. However, the iterative \(\LF\)  for \(\ISA\) has additional information
available initially (i.e., an $\SA$ value), which is what we exploit here.

\section{Building-Block Data Structures} \label{sec:4}
%In this section, we establish key building blocks for later sections.
\subsection{Constant-Time ${\LF}^H$ and $\LFm^H$ for a Fixed $H$}\label{sec_3.3}
In this section, we design structures supporting  $\LF^H$ and $\LFm^H$ operations in $O(1)$ time, where $H=O(\log_{\sigma} n)$ is a fixed integer parameter.
Let $\BWT'[0 \dd n)$ be a generalization of $\BWT$, where $\BWT'[i]$ consists of the last $H$ characters of the $i$-th row of the BWT-matrix $M$. Specifically,
\[
\BWT'[i] = M[i][n-H] \circ M[i][n-H+1] \circ \dots \circ M[i][n-1].
\]
Since each character of $\BWT'$ is a string of length $H$ over the original alphabet $[0 \dd \sigma)$, with each character represented using $\lceil \log \sigma\rceil$ bits, we can equivalently view $\BWT'[0 \dd n)$ as a string of length $n$ over the alphabet $[0 \dd 2^{H\lceil \log\sigma\rceil})$.
For any $0 \le i,j < n$, the order between $\mathrm{LF}^H[i]$ and $\mathrm{LF}^H[j]$ can be determined simply by comparing $\BWT'[i]$ and $\BWT'[j]$. In particular, we have
$$
\mathrm{LF}^H[j] < \mathrm{LF}^H[i]
\iff
\big(\BWT'[j] < \BWT'[i]\big)
\;\lor\;
\big(\BWT'[j] = \BWT'[i] \text{ and } j < i\big).
$$
By a simple induction on $H$ (with the base case established in Section~\ref{sec:BWT}) we obtain the following, where $c=\BWT'[i]$:
$$
\mathrm{LF}^H[i] = \mathrm{COUNT}_{\BWT'}[c] + \operatorname{rank}_{\BWT'}(i, c)-1.
$$

\begin{lemma} \label{lem_41}
Let $H = O(\log_{\sigma} n)$  be an integer parameter.
There exists an 
$O(nH\log\sigma)$-bit data structure that supports $\LF^H$ in $O(1)$ 
time. Therefore, it also supports $\LF^{tH}$ iteratively in $O(t)$ time.
\end{lemma}

\begin{proof}
Maintain $\BWT'$ in \emph{packed form}, together with the  structure from Lemma~\ref{lem_COUNTplusRANK} built over $\BWT'$. The total space is $n\log 2^{H\lceil \log\sigma\rceil} + O(n\log \log 2^{H\lceil \log\sigma\rceil}) \subseteq 
O(nH\log\sigma)$ bits. Given a query input $i$, we first obtain $c=\BWT'[i]$ by accessing $\BWT'$. Since $H=O(\log_{\sigma} n)$,   this access time is a constant. We then query the associated data structure with input $(i,c)$, which returns $\LF^H[i]$.
\end{proof}

\begin{corollary}
\label{cor_42}
Let $H=O(\log_{\sigma} n)$ be an integer parameter.
There exists an $O(n\log(H\log\sigma))$-bit data structure that supports $\LFm^H$ queries in $O(1)$ time, provided that the text is available (in packed form).
Therefore, it also supports $\LFm^{tH}$ iteratively in $O(t)$ time.
\end{corollary} 

\begin{proof}
Modify the structure in Lemma~\ref{lem_41} by not storing $\BWT'$.
Then while processing a query with input $(i,\SA[i])$, we recover $\BWT'[i]=T[((\SA[i]-H)\bmod n)\dd((\SA[i]-1)\bmod n)]$ from the text rather than accessing $\BWT'$, and proceed as before.
\end{proof}

\subsection{Constant-Time ${\LF}^H$ and $\LFm^H$ at $h$-Sampled Positions}\label{sec_3.4}
Here, we introduce an additional integer parameter 
$h \geq 1$ such that $h$ divides both $H$ and $n$. The goal is to achieve a space-efficient solution when $\LF^H$ and $\LFm^H$ queries are restricted to those $i$ for which $h$ divides $\SA[i]$, and we call such an $\SA[i]$ an \emph{$h$-sampled} (text) position. Note that if $\SA[i]$ is $h$-sampled, then $((\SA[i]-H) \bmod n)$ is also $h$-sampled.

\paragraph{Sampling the BWT-Matrix.}
Consider sampling the BWT-matrix $M$ by selecting only those rows $i$ for which the corresponding cyclic suffix $T_{\mathrm{SA}[i]}$ starts at a position divisible by $h$; that is, $\mathrm{SA}[i]$ is  $h$-sampled. Let $M'[0 \dd n/h)[0 \dd n)$ be the resulting matrix.  Based on $M'$, we next define a sampled suffix array $\SA'$, a bit vector $R$, and  $\BWT''$ (see Table~\ref{table:BWTsampled} for an illustration).

The sampled suffix array $\SA'[0 \dd n/h)$ consists of those entries in $\SA$ whose values are divisible by $h$, forming a  subsequence of $\SA$. The bit vector $R[0 \dd n)$ is defined by setting $R[i] = 1$ if and only if $\SA[i]$ belongs to $\SA'$ (equivalently, if $\SA[i]$ is divisible by $h$); this provides a mapping between positions in $\SA'$ and the sampled positions of $\SA$.
$\BWT''[0 \dd n/h)$ is a subsequence of $\BWT'$ (defined in Section~\ref{sec_3.3}) of length $n/h$, such that 
$\BWT''[i]$ 
consists of the last $H$ characters of the $i$-th row 
of $M'$. Specifically,
\[
\BWT''[i] = M'[i][n-H] \circ M'[i][n-H+1] \circ \dots \circ M'[i][n-1] \in [0 \dd 2^{H \lceil \log\sigma \rceil}).
\]

\begin{table}[h]
\centering
\renewcommand{\arraystretch}{1.2}

\begin{tabular}{c|c|c|l
@{\hspace{1.2cm}}c|c|l|c}
\hline
\multicolumn{4}{c}{\textbf{BWT-Matrix}} &
\multicolumn{4}{c}{\textbf{Sampled Matrix ($h=2, H=4$)}} \\
\hline
Row & SA & R & Cyclic suffix &
Row & SA$'$ & Sampled row & $\BWT''$ \\
\hline
0  & 11 & 0 & \$mississippi &
  &    &                      &       \\
1  & 10 & 1 & i\$mississipp &
0 & 10 & i\$missis\underline{sipp} & sipp \\
2  & 7  & 0 & ippi\$mississ &
  &    &                      &       \\
3  & 4  & 1 & issippi\$miss &
1 & 4  & issippi\$\underline{miss} & miss \\
4  & 1  & 0 & ississippi\$m &
  &    &                      &       \\
5  & 0  & 1 & mississippi\$ &
2 & 0  & mississi\underline{ppi\$} & ppi\$ \\
6  & 9  & 0 & pi\$mississip &
  &    &                      &       \\
7  & 8  & 1 & ppi\$mississi &
3 & 8  & ppi\$miss\underline{issi} & issi \\
8  & 6  & 1 & sippi\$missis &
4 & 6  & sippi\$mi\underline{ssis} & ssis \\
9  & 3  & 0 & sissippi\$mis &
  &    &                      &       \\
10 & 5  & 0 & ssippi\$missi &
  &    &                      &       \\
11 & 2  & 1 & ssissippi\$mi &
5 & 2  & ssissipp\underline{i\$mi} & i\$mi \\
\hline
\end{tabular}

\caption{Sampling of the BWT-matrix of  $T[0 \dd 12)=\text{mississippi\$}$ with $h=2$ and $H=4$.}
\label{table:BWTsampled}
\end{table}

% \begin{table}[h]
% \centering
% \renewcommand{\arraystretch}{1.2}

% \begin{tabular}{c|c|c|>{}l
% @{\hspace{1.5cm}}c|c|>{}l}
% \hline
% \multicolumn{4}{c}{\textbf{BWT Matrix}} &
% \multicolumn{3}{c}{\textbf{Sampled Matrix ($h=2$)}} \\
% \hline
% Row & SA & R & Circular suffix &
% Row & SA$'$ & Sampled row \\
% \hline
% 0  & 11 & 0 & \$mississippi &
%   &    &  \\
% 1  & 10 & 1 & i\$mississipp &
% 0 & 10 & i\$missis\underline{sipp} \\
% 2  & 7  & 0 & ippi\$mississ &
%   &    &  \\
% 3  & 4  & 1 & issippi\$miss &
% 1 & 4  & issippi\$\underline{miss} \\
% 4  & 1  & 0 & ississippi\$m &
%   &    &  \\
% 5  & 0  & 1 & mississippi\$ &
% 2 & 0  & mississi\underline{ppi\$} \\
% 6  & 9  & 0 & pi\$mississip &
%   &    &  \\
% 7  & 8  & 1 & ppi\$mississi &
% 3 & 8  & ppi\$miss\underline{issi} \\
% 8  & 6  & 1 & sippi\$missis &
% 4 & 6  & sippi\$mi\underline{ssis} \\
% 9  & 3  & 0 & sissippi\$mis &
%   &    &  \\
% 10 & 5  & 0 & ssippi\$missi &
%   &    &  \\
% 11 & 2  & 1 & ssissippi\$mi &
% 5 & 2  & ssissipp\underline{i\$mi} \\
% \hline
% \end{tabular}

% \caption{Sampling of the BWT matrix of the text $T[0 \dd 12)= \text{mississippi\$}$ with $h=2$ and $H=4$, where 
% $\BWT = \text{ipssm\$pissii}$ and 
% $\BWT'' =\text{sipp}\circ \text{miss}\circ \text{ppi\$} \circ\text{issi} \circ\text{ssis} \circ\text{i\$mi}$. }
% \label{table:BWTsampled}
% \end{table}

\begin{lemma} \label{lem:DS}
Let $h$ and $H$ be integer parameters (fixed at construction time) such that 
$H = O(\log_{\sigma} n)$ and $h$ divides both $H$ and $n$. 
There exists a data structure of size $O((H/h)n\log\sigma)$ bits that, given 
an index $i$ such that $\SA[i]$ is divisible by $h$ (i.e., $\SA[i]$ is $h$-sampled), computes $\LF^H[i]$ in 
$O(1)$ time. We call this data structure $DS(h,H)$.
%%%
Since $\SA[i]$ is $h$-sampled, $((\SA[i]-H) \mod n)$ is also $h$-sampled. Therefore, $\LF^{tH}[i]$ can be computed iteratively in $O(t)$ time.
 \end{lemma}

\begin{proof}
Maintain $\BWT''$ in  \emph{packed form}, the  structure  in Lemma~\ref{lem_COUNTplusRANK} built over $\BWT''$, and 
the bit vector $R$ as an
indexable dictionary. 
Then, given an input $i$, we perform the following steps.

\begin{enumerate}
    \item  Compute $i'=\rank_R(i,1)-1$ via a partial rank query ($R[i]=1$ since $\SA[i]$ is $h$-sampled), which is the row in $M'$ corresponding to 
    $T_{\SA[i]}$ (recall that we follow $0$-based indexing).

    \item Since $h$ divides $H$ and $n$, the cyclic suffix $T_{(\SA[i]-H) \bmod n}$ must also be present in $M'$, say at row $j'$. By reasoning analogous to that used earlier, we compute $j'$ as follows:
    \[
        j' = \COUNT_{\BWT''}[c] + \rank_{\BWT''}(i', c)-1,\quad \text{where } c = \BWT''[i'].
    \]
      As before, we first determine $c$ by accessing $\BWT''$, then compute $j'$ using the structure from Lemma~\ref{lem_COUNTplusRANK} over $\BWT''$.

    \item  From $j'$, find $j = \select_R(j'+1, 1)$, which is the row in $M$ corresponding to $T_{(\SA[i]-H) \bmod n}$.  

    \item Finally, return $j$ as $\LF^H[i]$.
\end{enumerate}
Each step takes $O(1)$ time, so the total time is also $O(1)$. 
The space required for storing $\BWT''$ is $(n/h)\log 2^{H\lceil \log\sigma\rceil}$ bits.
The auxiliary data structure from Lemma~\ref{lem_COUNTplusRANK} requires
$O((n/h) \log\log 2^{H\lceil \log\sigma\rceil})$ bits and the indexable dictionary over $R$ requires $O((n/h) \log h)$ bits. Therefore, the  total space is
$O((n/h) \cdot (H\log\sigma + \log (H\log \sigma) +
 \log h))  \subseteq O((H/h) n\log\sigma)$ bits.
\end{proof}

\begin{corollary} \label{cor:DS'}
Let $h$ and $H$ be integer parameters (fixed at construction time) such that 
$H = O(\log_{\sigma} n)$ and $h$ divides both $H$ and $n$.
There exists a data structure of size $O((n/h)  \log (H\log \sigma))$ bits that, given a pair $(i, \SA[i])$ such that $\SA[i]$ is divisible by $h$ (i.e., $\SA[i]$ is $h$-sampled), computes $\LFm^H[(i,\SA[i])]$ in 
$O(1)$ time, provided that the text is available (in packed form).
We call this data structure $DS'(h,H)$.

Since $\SA[i]$ is $h$-sampled, $((\SA[i]-H) \mod n)$ is also $h$-sampled. Therefore,  $\LFm^{tH}[(i, \SA[i])]$ can be computed iteratively in $O(t)$ time.
\end{corollary}

\begin{proof}
Modify the structure in Lemma~\ref{lem:DS} by not storing $\BWT''$.
The resulting space is 
$O((n/h) \cdot (\log(H\log \sigma) +
 \log h))  \subseteq O((n/h) \log(H \log\sigma))$ bits. 
While processing a query $(i,\SA[i])$,  recover $\BWT''[i']=T[((\SA[i]-H)\bmod n)\dd((\SA[i]-1)\bmod n)]$ from the text, and proceed  as before.
\end{proof}

\section{ISA Encodings with Doubly Logarithmic Query Time}\label{sec:5}

\paragraph{Adjusting the Text Length.}
Let $D=\max\{2,2^{\lfloor\log\log_{\sigma} n\rfloor}\}$, so that $D$ is a power of $2$ and $D=\Theta(\log_{\sigma} n)$.
Our first task is to modify the text, if necessary, so that its length is divisible by $D$. We achieve this by \emph{prepending} $d$ copies of an arbitrary character $c\neq \$$, where $d=D\lceil n/D\rceil-n<D$. This modification is purely conceptual, so the original text  in packed form need not be modified. 

\begin{lemma}
An \emph{SA/ISA} query on $T$ can be reduced to the corresponding query on the modified text in $O(1)$ time using an interface structure of space $O(n)$ bits.    
\end{lemma}
\begin{proof}
Let $T'[0 \dd n')$ be the modified text (here $n'=n+d$), and let $\SA'[0 \dd n')$ and $\ISA'[0 \dd n')$ denote its suffix array and inverse suffix array, respectively. Let $B[0 \dd n')$ be a bit vector such that $B[i]=1$ if and only if $\SA'[i]\geq d$, i.e., $T'[\SA'[i]\dd n')$ is also a suffix of $T$.
Naturally, for $i\in[0\dd n)$, $\SA[i]=\SA'[i']-d$, where $i'=\select_B(i+1,1)$, while $\ISA[i]=\rank_B(\ISA'[i+d],1)-1$.
Our interface data structure essentially stores $B$ in $O(n)$ bits and supports rank/select queries in $O(1)$ time.
\end{proof}

This interface structure, which does not affect the overall space complexity, allows us to design our auxiliary structure assuming that the text length is divisible by $D$. Therefore, throughout the remainder of this section, we continue to use $T$ and $n$ to denote the modified text and its length, respectively, for convenience.

\subsection{Reconstructing a CSA Space–Time Trade-off} \label{sec:CSA}
Analogous to the LF-mapping in the FM-index~\cite{DBLP:journals/jacm/FerraginaM05}, the central function of the CSA 
is the $\psi$ function~\cite{DBLP:journals/siamcomp/GrossiV05}, defined as $\psi[i]=\ISA[(\SA[i]+1)\bmod n]$. Since $\psi$ is analogous to the $\LF$-mapping, by replacing $\psi$ with $\LF$ and making appropriate modifications to the CSA framework, we can achieve the same space--time trade-off. To this end, this section presents the complete construction and analysis of (a slightly more general form of) one of the CSA space--time trade-offs. Then in the next section, we show how further modifications lead to the main result of this paper.

\paragraph{Data Structure.}
Let $\Delta \in [2,D]$ be a power of $2$ and 
$\ell=\lfloor\log_{\Delta}D\rfloor$;  recall that $D$ is a power of $2$
and $D=\Theta(\log_\sigma n)$.  We also set the sampling parameter for $\ISA$ to $g=\tau D$, where
$\tau$ is an integer-valued parameter such that 
$1 \leq \tau =O(\log \sigma)$.
The data structure stores the sampled $\ISA$, using $O((n/g)\log n)\subseteq O((n/\tau)\log\sigma)$ bits, together with the structures $DS(\cdot,\cdot)$ from Lemma~\ref{lem:DS} for all pairs of parameters listed below whose second component is at most $D$. Note that the choice of $D$ is based on the fact that $D$ characters of the
text can be read in $O(1)$ time.

\begin{align*}
&(1,1), (1,2), (1,3), \dots, (1, \Delta-1) \\
&(\Delta, \Delta), (\Delta, 2\Delta), (\Delta, 3\Delta), \dots, (\Delta, (\Delta-1)\Delta) \\
&(\Delta^2, \Delta^2), (\Delta^2, 2\Delta^2), (\Delta^2, 3\Delta^2), \dots, (\Delta^2, (\Delta-1)\Delta^2) \\
&\vdots \\
&(\Delta^{\ell}, \Delta^{\ell}), (\Delta^{\ell}, 2\Delta^{\ell}), (\Delta^{\ell}, 3\Delta^{\ell}), \dots, (\Delta^{\ell}, (\Delta-1)\Delta^{\ell})
\end{align*}

Additionally, we maintain $DS(D,D)$. The total space for  all $DS(\cdot,\cdot)$ structures is
\[
n \log \sigma+\sum_{k=0}^{\ell}  \sum_{j=1}^{\Delta-1} j  n \log \sigma \leq \ell \Delta^2 \cdot n \log \sigma = O(
n \Delta^2   \log\sigma \cdot \log_{\Delta} \log_{\sigma}n ) \text{ bits.}
\]

\paragraph{The Query Algorithm.}
Recall that our input is $(i, \SA[i], t)$ for some $t < g$, where $\SA[i]$  is divisible by $g$, and
the task is to compute $\LF^t[i]$. 
Let $\alpha \geq 0$ and $\beta \in [0, D)$ be integers, such that $t=\alpha D+\beta$.
Compute the base-$\Delta$ representation of $\beta$, which is   
 the unique sequence of digits $0 \le t_0, t_1, \dots, t_\ell < \Delta$ such that 
$
\beta = t_0 + t_1 \Delta + t_2 \Delta^2 + \dots + t_\ell \Delta^\ell
$  (this step can be done in $O(\ell)$ time).
Next compute $i'_{\ell+1}=\LF^{\alpha D}[i]$, and then iteratively compute
\[
i'_k=\LF^{t_k\Delta^k}[i'_{k+1}],
\quad k=\ell,\ell-1,\ldots,0,
\]
until we obtain $i'_0 =\LF^{\alpha D+\beta}[i]$ (the desired answer).

We now describe the implementation details. 
Since $\SA[i]$ is divisible by $g$ and therefore also by $D$, the position $\SA[i]$ is $D$-sampled. Hence, we can query the data structure $DS(D,D)$ to compute $i'_{\ell+1}=\LF^{\alpha D}[i]$ by applying $\LF^D$ exactly $\alpha$ times. Note that if  $\alpha=0$, we can simply set  $i'_{\ell+1}=i$.

To compute $i'_{\ell}$, observe that $\SA[i'_{\ell+1}]=(\SA[i]-\alpha D)\bmod n$ is divisible by $\Delta^{\ell}$, and hence $i'_{\ell+1}$ is $\Delta^{\ell}$-sampled. Thus, if $t_{\ell}\neq 0$, we query $DS(\Delta^{\ell},t_{\ell}\Delta^{\ell})$ on $i'_{\ell+1}$ to obtain
$i'_{\ell}=\LF^{t_{\ell}\Delta^{\ell}}[i'_{\ell+1}]$.
If $t_{\ell}=0$, we simply set $i'_{\ell}=i'_{\ell+1}$.

To compute $i'_k$ for each $k=\ell-1,\ell-2,\ldots,0$, observe that $\SA[i'_{k+1}]=(\SA[i]-\alpha D-\sum_{j=k+1}^{\ell}t_j\Delta^j)\bmod n$, which is divisible by $\Delta^k$, implying that $\SA[i'_{k+1}]$ is $\Delta^k$-sampled. Thus, if $t_k \neq 0$, we  query $DS(\Delta^k,t_k\Delta^k)$ with input $i'_{k+1}$ to obtain $i'_k=\LF^{t_k\Delta^k}[i'_{k+1}]$.
If $t_{k}=0$, we simply set $i'_{k}=i'_{k+1}$.

The overall time can be bounded as 
$\alpha+\ell = O(g/D+\log_{\Delta} D)   \subseteq O(g/\log_{\sigma} n+\log_{\Delta}\log_{\sigma} n)$. Finally, the result below follows by substituting $g=\tau D$.

\begin{lemma}
There exists an ISA representation using 
$O((n/\tau)\log \sigma+ n\Delta^2  \log\sigma \cdot \log_{\Delta}\log_{\sigma} n)$ bits
of space and supporting queries in 
$O(\tau+\log_{\Delta}\log_{\sigma} n)$ time.
\end{lemma}

Setting $\tau=1$ and $\Delta=2$ yields the second CSA trade-off in Table~\ref{table_CSA}. 
Alternatively, setting $\tau=1$ and $\Delta=\max\{2,\log_{\sigma}^{\epsilon/2} n\}$, rounded up to the next power of $2$, yields the third result in Table~\ref{table_CSA}.

\vspace{-4mm}
\subsection{Final Data Structures}
We modify the solution from Section~\ref{sec:CSA} as follows.  
For each of the listed pairs, instead of using the LF-based structure $DS(\cdot,\cdot)$ described in Lemma~\ref{lem:DS}, we use the $\LFm$-based structure $DS'(\cdot,\cdot)$ described in Corollary~\ref{cor:DS'}.  
Note that there are at most $(\ell+1)(\Delta-1)+1$ such pairs, and hence we maintain the same number of versions of $DS'(\cdot,\cdot)$.  
The advantage of our BWT-free (text-based) approach is that 
        \emph{ the text itself is stored only once across all these auxiliary structures}.

\paragraph{Query Algorithm.}
We follow essentially the same procedure as before to compute
\(i'_{\ell+1}, i'_{\ell}, \dots, i'_0\). However, each \(\LF\) operation is
replaced by the corresponding \(\LFm\) operation by leveraging the fact that
the SA values, i.e., \( \SA[i'_{\ell+1}], \SA[i'_{\ell}], \dots, \SA[i'_0]\), are readily available.
Specifically, a query with input \(q\) to \(DS(a,b)\) for computing
\(\LF^b[q]\) is replaced by a query with input \((q,\SA[q])\) to \(DS'(a,b)\)
for computing \(\LFm^b[(q,\SA[q])]\). This is possible because the initial
input is \((i,\SA[i])\), rather than only \(i\).
%The query time remains unchanged.
The query time remains $O(\tau+\log_{\Delta}\log_{\sigma} n)$, as before.

\paragraph{Bounding the Space.}
The combined space of  all $DS'(\cdot,\cdot)$ structures (in bits) is:
\[
\begin{aligned}
\sum_{k=0}^{\ell} \sum_{j=1}^{\Delta-1} 
\frac{n}{\Delta^k} \log (j \Delta^k \log \sigma) +\frac{n}{D} \log (D \log\sigma)
&\leq \sum_{k=0}^{\ell} \frac{n}{\Delta^k} \cdot \Delta \log (\Delta^{k+1} \log \sigma) +n\log\log\sigma \\
&\leq  
n \Delta \log \Delta \sum_{k=0}^{\ell} \frac{k+1}{\Delta^k}+n\log\log\sigma \Big(1+ \Delta \sum_{k=0}^{\ell} \frac{1}{\Delta^k}  \Big)
 \\
&= O\big(n \Delta \log (\Delta  \log \sigma)\big).
\end{aligned}
\]
Recall that we store only the original text, while conceptually treating it
as modified.
The sampled ISA takes $O((n/\tau)\log\sigma)$ bits, as before.

\begin{lemma}\label{lemma:ISA}
There exists a structure of size
$O((n/\tau)\log\sigma)+O(n\Delta\log(\Delta\log\sigma))$ bits that, together with the text in packed form, supports $\ISA$ queries in $t_{\ISA}=O(\tau+\log_{\Delta}\log_{\sigma} n)$ time.
\end{lemma}
We now show how to obtain the $\ISA$ results in Theorem~\ref{thm_main} using the above lemma.
For the first $\ISA$ result, we set $\Delta=\log^{\epsilon}\sigma/\log\log\sigma$, rounded up to the next power of $2$.\footnote{More precisely, if $\log^{\epsilon}\sigma/\log\log\sigma<D$, we set $\Delta=\max\{2,\log^{\epsilon}\sigma/\log\log\sigma\}$, rounded up to the next power of $2$; otherwise, we set $\Delta=\max\{2,\log^{\epsilon}_{\sigma}n\}$, rounded up to the next power of $2$ and capped at $D$.}
For the second $\ISA$ result,  simply set $\Delta=2$.
The third ISA result is more involved, which we establish in Lemma~\ref{lemma:ISA_third}.

\begin{lemma}\label{lemma:ISA_third}
There exists a structure of size $O((n/\tau)\log\sigma)+O(n\log\log\log\sigma)$ bits that, with the text in packed form, supports $\ISA$ queries in $t_{\ISA}=O((\tau+\log\log\log \sigma)\log\log\sigma+ \log\log n)$ time.
\end{lemma}

\begin{proof}
% The idea is to replace the MMPHF used in the proof of Lemma~\ref{lemma:ISA} with a more space-efficient variant. The MMPHF used there requires $O(m\log\log(|B|/m))$ bits and supports queries in $O(1)$ time, where $m$ is the number of $1$-bits in the underlying bit string $B$. The  more space-efficient variant requires only $O(m\log\log\log(|B|/m))$ bits, at the cost of increasing the query time to $O(\log\log(|B|/m))$~\cite{DBLP:conf/soda/BelazzouguiBPV09,DBLP:conf/alenex/BelazzouguiBPV09,DBLP:journals/talg/BelazzouguiN14}.\footnote{The lower-bound results in ~\cite{DBLP:conf/soda/AssadiFK23,DBLP:conf/cpm/Kosolobov24} show that this space bound is asymptotically optimal.} 
While designing $DS'(h,H)$ in Corollary~\ref{cor:DS'}, we use an 
$O(m\log\log(|B|/m))$-bit
MMPHF that  supports queries in $O(1)$ time, where
$m$ denotes the number of $1$-bits in the underlying bit string $B$.
We replace this MMPHF with a more space-efficient variant that requires
$O(m\log\log\log(|B|/m))$ bits and supports queries in
$O(\log\log(|B|/m))$ time~\cite{DBLP:conf/soda/BelazzouguiBPV09,DBLP:conf/alenex/BelazzouguiBPV09,DBLP:journals/talg/BelazzouguiN14,DBLP:conf/soda/AssadiFK23,DBLP:conf/cpm/Kosolobov24}.
Consequently, the space and query-time bounds in Lemma~\ref{lem_COUNTplusRANK} become
$O(n\log\log\log\pi)$ bits and $O(\log\log\pi)$, respectively.
Accordingly, the space and query-time bounds for $DS'(h,H)$ in Corollary~\ref{cor:DS'} become
$O(n/h \cdot(\log\log(H\log\sigma)+\log h))$ bits and
$O(\log(H\log\sigma))$, respectively.
For convenience, we refer to this modified structure as $DS''(h,H)$.

We modify the construction of Lemma~\ref{lemma:ISA} as follows. Fix $\Delta=2$. Then $\ell=\log_{\Delta}D$ exactly, with no need for floors, and the data structure consists of $DS'(2^k,2^k)$ for $k=0,1,2,\ldots,\ell$. Let $k'$ be the integer satisfying
$
2^{k'-1}<\log\log\sigma\leq 2^{k'},
$
so that $k'<1+\log\log\log\sigma$.
Set $k'' =\min\{k', \ell\}$.
For $k=0,1,2,\ldots,k''$, we replace $DS'(2^k,2^k)$ with $DS''(2^k,2^k)$. Therefore, the total space used by these auxiliary data structures is

\[
\begin{aligned}
&\sum_{0 \leq k \leq k''}\frac{n}{2^k}\log\log(2^k\log\sigma)
+\sum_{k''< k \leq \ell}\frac{n}{2^k}\log(2^k\log\sigma)
+\sum_{0 \leq k\leq \ell}\frac{n}{2^k} \cdot k\\
&\qquad=
O\left(
n\log\log\log\sigma
\left(1+\sum_{0< k\leq k''}\frac{k}{2^k}\right)
+
n \cdot \frac{\log\log\sigma}{2^{k''}}
\sum_{1 \leq r\leq \ell-k''}\frac{k''+r}{2^r}
+n\right)\\
&\qquad=
O\left(
n\log\log\log\sigma
+
nk''\sum_{1 \leq r\leq \ell}\frac{r}{2^r}+n
\right)\\
&\qquad=
O\left(n\log\log\log\sigma\right) \text{ bits}.
\end{aligned}
\]
%%%
\\
To bound the query time, we distinguish two cases.

\paragraph{Case 1: $k''=k'<\ell$.}
For $k=\ell$, the corresponding data structure has $O(1)$ query time. Hence, the query time is 

$$
\begin{aligned}
\alpha+\ell+\sum_{k=0}^{k''}\log(2^k\log\sigma)
& \leq 
\tau+\ell+(k''+1)\log(2^{k''}\log\sigma)
\\
&=O\!\left(
\tau+\log\log_{\sigma}n
+\log\log\sigma\,\log\log\log\sigma
\right).
\end{aligned}
$$

\paragraph{Case 2: $k''=\ell\leq k'$.}
In this case, we have $DS''(2^k,2^k)$ for every $k\in[0,\ell]$, and
$\ell=O(\log\log\log\sigma)$. Therefore, the query time is 

$$
\begin{aligned}
\alpha\log(2^\ell\log\sigma)
+\sum_{k=0}^{\ell}\log(2^k\log\sigma)
&\leq
(\alpha+\ell+1)\log(2^\ell\log\sigma)\\
&=O\!\left(
(\tau+\log\log\log\sigma)\log\log\sigma
\right).
\end{aligned}
$$

% To bound the query time, we consider two cases: (i) $k''= k' < \ell$ and (ii) $k'' = \ell \leq k'$. 

% \begin{enumerate}
%     \item In the first case, we have data structures with $O(1)$ query time for $k=\ell$, resulting in the following time complexity.
    
% $$
% \begin{aligned}
% \alpha+\ell+\sum_{k=0}^{k''}\log(2^k\log\sigma)
% &=O\!\left(\tau+\ell+k''\log\log\sigma+(k'')^2\right)\\
% &=O\!\left(\tau+\log\log_{\sigma}n
% +\log\log\sigma\log\log\log\sigma\right).
% \end{aligned}
% $$

% \item 
% In the second case, we have $DS''(2^k, 2^k)$ for all $k \in [0,\ell]$, and $\ell =O(\log\log\log\sigma)$, resulting in the following time complexity:

% $$
% \begin{aligned}
% \alpha \log(2^{\ell} \log\sigma) +\sum_{k=0}^{\ell}\log(2^k\log\sigma)
% & \leq (\alpha +\ell+1)  \log(2^{\ell} \log\sigma) \\
% &=O\!\left((\tau+\log\log\log \sigma) \log\log {\sigma}
% \right).
% \end{aligned}
% $$
% \end{enumerate}
Therefore, combining both cases, we obtain the following overall time complexity:
$$O((\tau+\log\log\log \sigma)\log\log\sigma+ \log\log n).$$
This completes the proof.
\end{proof}

\paragraph{Supporting Suffix Array Queries.}
Consider the following result by Munro {\it et al.}~\cite{DBLP:journals/tcs/MunroRRR12}: suppose a permutation $\varphi$ of $[0\dd n)$ is stored as a black box supporting $\varphi[i]$ queries, i.e., reporting the $i$th entry of the permutation. Then, by maintaining an auxiliary structure of size $O((n/s)\log n)$ bits, $\varphi^{-1}[j]$ queries, i.e., finding the position of $j$ in $\varphi$, can be answered using $O(s)$ $\varphi[\cdot]$ queries. We set $\varphi=\ISA$, so that $\varphi^{-1}=\SA$, and choose $s=\tau \lceil \log_{\sigma} n \rceil $. The auxiliary structure uses $O((n/s)\log n)=O((n/\tau)\log\sigma)$ bits, and each $\SA$ query can be answered using $O(\tau\log_{\sigma}n)$ $\ISA$ queries.
This completes the proof of the claimed SA query bound of Theorem~\ref{thm_main}.

\section{Query-Time Lower Bound in Terms of Auxiliary Space}\label{sec:LB}
Both the suffix array and the inverse suffix array require $\Omega(n\log\sigma)$ bits in the worst case. Given either the suffix array or the inverse suffix array of $T$, together with the number of occurrences of each character in $[0\dd\sigma)$, we can reconstruct $T$. The latter information can be encoded in $o(n\log\sigma)$ bits; for example, in $O(\sigma\log(n/\sigma))$ bits using Elias--Fano representation (see~\cite{kempa2026cellSALB} for a formal proof). Hence, the claimed lower bound follows. A more interesting question is therefore: 
\begin{center}
\emph{How efficiently can various queries be supported with bounded auxiliary space?}
\end{center}
This question has been studied extensively for $\LCE$ queries~\cite{LCE1,LCE2,LCE3,LCE4}. The best-known result offers, for a parameter $\tau$, $O(\tau)$ query time using $O((n/\tau)\log n)$ bits of auxiliary space.
In contrast, to the best of our knowledge, this auxiliary-space perspective had not previously been studied for $\SA$ and $\ISA$ queries. Having established upper bounds, we now turn to lower bounds.

Our lower bound is proved in the \emph{cell-probe model} with cell size $\Theta(\log n)$ bits, where computation is free and query time is measured by the number of memory cells probed by the query algorithm~\cite{DBLP:conf/stoc/FredmanS89}. In particular, the lower bound also applies to the word RAM model.

\begin{theorem}
In the cell-probe model with cell size $\Theta(\log n)$ bits, suppose the text
$T[0\dd n)\in[0\dd\sigma)^n$, where $\sigma\leq n$, is stored in packed
form, together with a deterministic auxiliary data structure of
$S(n,\sigma)=\Omega(n)$ bits. If $\SA$ queries are supported in
$t_{\SA}$ probes, then
\[
t_{\SA}
=
\Omega\left(
\frac{\log\sigma}
{S(n,\sigma)/n+\log^2\sigma/\log n}
\right).
\]
\end{theorem}

\begin{proof}
For constant $\sigma$, the right-hand side is $O(1)$, and the claimed
lower bound is trivial. Hence, for the remainder of the proof, assume
$\sigma\geq 4$.
We reduce from permutation indexing. Let $\rho$ be a permutation of
$[0\dd m)$, and let $t_{\rho}$ and $t_{\rho^{-1}}$ denote the query
times for $\rho[\cdot]$ and $\rho^{-1}[\cdot]$, respectively.
Golynski's lower bound states that any representation using
$\log (m!)+R(m)$ bits satisfies
$t_{\rho}t_{\rho^{-1}}=\Omega(m\log m/R(m))$~\cite{DBLP:conf/soda/Golynski09}.\footnote{More precisely, Golynski's bound assumes
$\max\{t_{\rho},t_{\rho^{-1}}\}<m^{1-\epsilon}$ and
$\min\{t_{\rho},t_{\rho^{-1}}\}=O(\log m/\log\log m)$ for a fixed
constant $\epsilon>0$. Since $t_{\rho}=O(1)$, the second condition is
immediate. For the first, take $\epsilon=1/2$: if
$t_{\rho^{-1}}=\Omega(\sqrt m)$, the claimed bound already follows,
since $S(n,\sigma)=\Omega(n)$ and $\sigma\leq n$ imply
$\log\sigma/(S(n,\sigma)/n+\log^2\sigma/\log n)
=O(\sqrt{\log n})=o(\sqrt m)$; otherwise,
$t_{\rho^{-1}}<\sqrt m$ and Golynski's bound applies.}

Let $b=\lfloor\log\sigma\rfloor-1$. Represent each $\rho[i]$ by its
$\lceil\log m\rceil$-bit binary representation and prepend the minimum number of  $1$s so that its length is divisible by $b$. Since the same prefix
is added to every value, their lexicographic order is preserved. Split
each padded representation into $b$-bit blocks and prepend a $1$ to
each block. Each resulting block has $\lfloor\log\sigma\rfloor$ bits
and therefore represents a character in $[0\dd\sigma)$; moreover, it
is distinct from the all-zero character $\$$. Concatenate the encodings
of $\rho[0],\ldots,\rho[m-1]$ and append $\$$, making $\$$ the unique
smallest character of $T$.

Let $q=\lceil\lceil\log m\rceil/b\rceil$ be the number of characters
encoding each permutation value. Then
$n=mq+1=\Theta(m\log m/\log\sigma)$. Since $\sigma\leq n$, we have
$\log m=\Theta(\log n)$ and
$m=\Theta(n\log\sigma/\log n)$.
A query $\rho[i]$ can be answered in $O(1)$ probes by reading and
decoding its $q$ consecutive characters, which occupy $O(1)$ machine
words. Moreover, for any $i\neq j$, the suffixes starting at the
encodings of $\rho[i]$ and $\rho[j]$ differ within their first $q$
characters, and their lexicographic order is exactly the order of
$\rho[i]$ and $\rho[j]$. Thus, by storing a bitvector that marks in
the suffix array the suffixes starting at encoding boundaries,
$\rho^{-1}[i]$ can be answered by selecting the $(i+1)$th marked suffix,
performing one $\SA$ query, and converting the returned text position
to its encoding index. The bitvector takes $O(n)$ bits and supports
select in $O(1)$ probes. Hence $t_{\rho}=O(1)$ and
$t_{\rho^{-1}}=O(t_{\SA})$.

The packed text uses $m\log m+O(m\log\sigma+n)$ bits, so
$
R(m)=S(n,\sigma)+O(n\log^2\sigma/\log n),
$
using $m=\Theta(n\log\sigma/\log n)$ and $S(n,\sigma)=\Omega(n)$.
Applying Golynski's lower bound gives
$
t_{\SA}
=
\Omega\!\left(
{m\log m}/{(S(n,\sigma)+n\log^2\sigma/\log n)}
\right),
$
which simplifies to
$
\Omega\!\left(
{\log\sigma}/{(S(n,\sigma)/n+\log^2\sigma/\log n)}
\right).$
\end{proof}

A simple interpretation of this result is that with $O((n/\tau)\log\sigma)$ bits of auxiliary space, $\SA$ queries require $\Omega(\tau)$ time for $\tau=O(\min\{\log\sigma,\log_{\sigma}n\})$.

 \section{Concluding Remark}\label{sec:6:conc}
% A natural follow-up question is whether even faster $\ISA$ queries can be achieved within succinct/compact space. Kempa and Kociumaka recently answered this question positively in an arXiv paper, showing that $O(1)$-time $\ISA$ queries are possible in compact space for all alphabet sizes~\cite{kempa2026constantISA}. In another recent arXiv paper, they showed that $O(1)$-time $\SA$ queries are not possible in compact space~\cite{kempa2026cellSALB}. 
% This unconditionally confirms our conjecture that the query complexities of $\SA$ and $\ISA$ can be separated in compact space.

 A natural follow-up question is whether even faster $\ISA$ queries can be achieved within succinct/compact space. Kempa and Kociumaka recently answered this question positively in an arXiv paper, showing that $O(1)$-time $\ISA$ queries are possible in compact space for all alphabet sizes~\cite{kempa2026constantISA}. In another recent arXiv paper, they showed that $O(1)$-time $\SA$ queries are not possible in compact space~\cite{kempa2026cellSALB}. 
 This unconditionally confirms that the query complexities of SA and ISA can be separated in compact space.
 We conclude by leaving open the following question: can the current $O(n\log\log\log \sigma)$-bit auxiliary space for suffix trees/arrays be reduced further while retaining polylogarithmic query times?

\paragraph{Acknowledgments.}
 This research was supported in part by U.S. National Science Foundation (NSF) Awards CCF-2316691 and CCF-2315822. The author would like to thank the reviewers of FOCS 2026 for their valuable feedback and constructive comments.

\paragraph{AI Disclosure.} All research ideas, technical contributions, proofs, and results presented in this paper were developed by the author. Generative AI tools were used solely to assist with proofreading, improving the presentation of the manuscript, and identifying potential typographical, grammatical, stylistic, or logical issues. The author independently evaluated all AI-generated suggestions and takes full responsibility for the content and correctness of the paper.

	\bibliographystyle{alphaurl} 
	\bibliography{main}

\end{document}